\documentclass[cleveref,a4paper,USenglish,numberwithinsect]{lipics-v2021}
\usepackage{cite} % sort cite numbers
\usepackage{amsmath}
\usepackage{amssymb}  
\usepackage{xspace}

\usepackage{tikz}
\usepackage{todonotes}
\tikzstyle{every picture} = [>=latex]
\usetikzlibrary{calc,arrows,decorations.markings}

\usepackage[appendix=inline]{apxproof}
\ifx\proof\inlineproof
  \let\apxmark\relax  
   
\else
  \def\apxmark{\textcolor{lipicsGray}{\sf\bfseries$\!$(*)\,}}
  
\fi
\newtheoremrep{proposition2}[theorem]{Proposition}
\newtheoremrep{theorem2}[theorem]{Theorem}
\newtheoremrep{corollary2}[theorem]{Corollary}
\newtheoremrep{lemma2}[theorem]{Lemma}
\newtheoremrep{claim2}[theorem]{Claim}
\newtheoremrep{property2}[theorem]{Property}

\newcommand\FO{FO\xspace}
\newcommand\FPT{{\sf FPT}\xspace}
\def\mx#1{\mbox{\boldmath{$#1$}}}
\def\ca#1{{\cal#1}}
\def\cf#1{{\EuScript#1}}
\def\prebox#1{\mathop{\text{\sl #1}}}

\hideLIPIcs
\nolinenumbers

\title{Twin-width and Limits of Tractability of FO Model Checking on Geometric Graphs}
\titlerunning{Twin-width and Limits of Tractability of FO Model Checking}
\author{Petr Hlin{\v e}n\'y}{Masaryk University, Brno,
  Czechia}{hlineny@fi.muni.cz}{https://orcid.org/0000-0003-2125-1514}{}
\author{Filip Pokr\'yvka}{Masaryk University, Brno,
  Czechia}{xpokryvk@fi.muni.cz}{https://orcid.org/0000-0003-1212-4927}{}
\authorrunning{P.\ Hlin\v{e}n\'y and F.\ Pokr\'yvka}
\Copyright{Petr Hlin\v{e}n\'y and Filip Pokr\'yvka}


\funding{Supported by the {Czech Science Foundation}, project no.~{20-04567S}.
The Dagstuhl seminar 21391 {\em Sparsity in Algorithms, Combinatorics and Logic} (September 2021) significantly contributed to part of these results.}


\begin{document}

\maketitle
\begin{abstract}
The complexity of the problem of deciding properties expressible in FO logic on graphs -- the FO model checking problem (parameterized by the respective FO formula),
is well-understood on so-called sparse graph classes, but much less understood on hereditary dense graph classes.
Regarding the latter, a recent concept of twin-width [Bonnet et al., FOCS 2020] appears to be very useful.
For instance, the question of these authors [CGTA 2019] about where is the exact limit of fixed-parameter tractability
of \FO model checking on permutation graphs has been answered by Bonnet et al.~in 2020 quite easily, using the newly introduced twin-width.
We prove that such exact characterization of hereditary subclasses with tractable FO model checking naturally extends from permutation to circle graphs (the intersection graphs of chords in a circle).
Namely, we prove that under usual complexity assumptions, FO model checking of a hereditary class of circle graphs is in \FPT if and only if the class excludes some permutation graph.
We also prove a similar excluded-subgraphs characterization for hereditary classes of interval graphs with FO model checking in \FPT,
which concludes the line a research of interval classes with tractable FO model checking started in [Ganian et al., ICALP 2013].
The mathematical side of the presented characterizations -- about when subclasses of the classes of circle and permutation graphs have bounded twin-width,
moreover extends to so-called bounded perturbations of these classes.
\end{abstract}

\section{Introduction}
%%%%%%%%%%%%%%%%%%%%%%%%%%%%%%%%%%%%%%%%%%%%%%%%%%%%%%%%%%%%%%%%%%%%%%%

So-called algorithmic meta-theorems receive considerable attention in theoretical computer science.
One particular challenge in this direction is about tractability of the model checking problem for {\em first-order logic} (\FO)
on graphs -- given a graph $G$ and an \FO formula $\phi$, the task to decide whether $G$ satisfies $\phi$ (written as $G\models\phi$).
This task is trivially solvable in time $|V(G)|^{\mathcal{O}(|\phi|)}$.
``Efficient solvability'', or tractability, hence in this context often means {\em fixed-parameter tractability} (\FPT);
that is, solvability in time $f(|\phi|)\cdot|V(G)|^{\mathcal{O}(1)}$ for some computable function~$f$.  

While for the \FO model checking problem on sparse graph classes we know a full answer---Grohe, Kreutzer and Siebertz~\cite{gks14}
proved that \FO model checking is \FPT on {\em nowhere dense} graph classes, while it is intractable (under usual complexity assumptions)
on all monotone somewhere dense classes---much less is known about the problem on hereditary dense graph classes.
Research in this direction has recently received a new strong stimulus in the form of the notion of {\em twin-width},
introduced in 2020 by Bonnet, Kim, Thomass{\'{e}} and Watrigant~\cite{DBLP:conf/focs/Bonnet0TW20}.

The basic definition of twin-width, in a condensed form, is as follows.

A \emph{trigraph} is a simple graph $G$ in which some edges are marked as {\em red}, and we then naturally speak about \emph{red neighbors} and \emph{red degree} in~$G$. We denote the set of red neighbors of a vertex $v$ by $N_r(v)$.
For a pair of (possibly not adjacent) vertices $x_1,x_2\in V(G)$, we define a \emph{contraction} of the pair $x_1,x_2$ as the operation
creating a trigraph $G'$ which is the same as $G$ except that $x_1,x_2$ are replaced with a new vertex $x_0$ whose full neighborhood %(including red edges)
is the union of neighborhoods of $x_1$ and $x_2$ in $G$, that is, $N(x_0)=(N(x_1)\cup N(x_2))\setminus\{x_1,x_2\}$, and
the red neighbors $N_r(x_0)$ of $x_0$ inherit all red neighbors of $x_1$ a`nd of $x_2$ and those in $N(x_1)\bigtriangleup N(x_2)$,
that is, $N_r(x_0)=\big((N_r(x_1)\cup N_r(x_2))\setminus\{x_1,x_2\}\big)\cup\big(N(x_1)\bigtriangleup N(x_2)\big)$.
A \emph{contraction sequence} of a trigraph $G$ is a sequence of successive contractions turning~$G$  into a single vertex, and its \emph{width} is the maximum red degree of any vertex in any trigraph of the sequence.
The \emph{twin-width} is the minimum width over all possible contraction sequences (where for an ordinary graph, we start with the same trigraph having no red edges).

This new notion has already found many very interesting applications,
which span from efficient parameterized algorithms and algorithmic metatheorems, through finite model theory, to classical combinatorial questions.
See the (still growing) series of follow-up papers
\cite{DBLP:conf/soda/BonnetGKTW21,DBLP:conf/icalp/BonnetG0TW21,DBLP:journals/corr/abs-2102-03117,DBLP:journals/corr/abs-2102-06880,DBLP:journals/corr/abs-2112-08953,DBLP:conf/soda/BonnetKRT22,DBLP:journals/corr/abs-2204-00722}.

\smallskip
In particular, graph classes of bounded twin-width have \FO model checking in \FPT, assuming the input graph comes with a suitable contraction sequence.
The input assumption is crucial, since finding the exact value of twin-width is para-NP-hard~\cite{DBLP:journals/corr/abs-2112-08953}, and we so far have no efficient approximation or asymptotic algorithm for it.
Generally, bounded twin-width (of hereditary graph classes) does not characterize tractability of \FO model checking;
a prominent example are the graphs of bounded degree with tractable \FO model checking~\cite{seese96} and unbounded twin-width~\cite{DBLP:conf/soda/BonnetGKTW21}.
However, for some natural types of graphs, such an exact characterization is true.
For instance, \cite{DBLP:conf/focs/Bonnet0TW20} have proved this for the permutation graphs;
a hereditary class of permutation graphs has \FO model checking in \FPT if and only if it has bounded twin-width, which is if and only if it excludes some permutation graph.
This provided a full answer to a question in an earlier paper by these authors~\cite{DBLP:journals/comgeo/HlinenyPR19}.

Directly inspired by the permutation-graph case, we extend the result to circle graphs (the intersection graphs of chords in a circle)
-- they again have bounded twin-width and \FO model checking in \FPT, if and only if they exclude some fixed permutation graph.
The characterization of bounded twin-width (but not the complexity part) extends even to a certain asymptotic generalization of circle graphs.
% Namely, in these hereditary classes the property of bounded twin-width coincides again with absence of some permutation graph.
We also investigate interval graphs, which represent one of the early examples of dense graph classes on which \FO model checking was systematically studied~\cite{ghkost13}.
For them we prove an exact characterization of tractable \FO model checking analogous to the circle-graph case, 
and again extend the bounded twin-width characterization to an asymptotic generalization of interval graphs.
This answers some of the questions left open in \cite{ghkost13}, and in \cite{DBLP:conf/focs/GajarskyHLOORS15,DBLP:journals/corr/GajarskyHLOORS15} specifically for interval graphs.

It is worth to note that in both directions which we study here, boundedness of twin-width is more or less explicitly related to expressibility of some permutations in the classes.
In related direction, Bonnet et al.~\cite{DBLP:journals/corr/abs-2102-06880} proved that a graph class is of bounded twin-width,
if and only if the class is an \FO transduction of a proper class of permutations.
This general asymptotic result is not directly comparable with the characterizations we prove~here.
On the other hand, another current paper of Bonnet et al.~\cite{DBLP:journals/corr/abs-2204-00722}
independently deals with related aspect of twin-width of interval graphs, and we discuss this below in details.

\subsection{Outline of the paper}
\begin{itemize}\parskip2pt
\item In Section \ref{sec:prel}, we give an overview of the necessary concepts from graph theory and logic;
	namely about intersection graphs (permutation, interval and circle graphs),
	the twin-width measure and its basic properties, and about \FO interpretations and transductions.

\item In Section \ref{sec:ilrepres}, we present a unified handling of interval and circle intersection representations,
	and prove technical claims for use within the coming main results.

\item In Section \ref{sec:circle}, we focus on the circle graphs, and prove the following.
\begin{itemize}
\item(\cref{thm:circlechar}) In a hereditary class $\cf C$ of circle graphs, these are equivalent: that (i) $\cf C$ is of bounded twin-width,
	(ii) $\cf C$ excludes some (fixed) permutation graph, and (iii) \FO model checking on $\cf C$ is in \FPT (under the assumption of ETH).
\item(\cref{thm:circlepertu}) In a hereditary class $\cf C$ of graphs which are ``bounded perturbations'' of circle graphs, these are equivalent:
	that (i) $\cf C$ is of bounded twin-width, and (ii) $\cf C$ excludes some (fixed) permutation graph.
\end{itemize}

\item In Section \ref{sec:interval}, we deal with interval graphs.
For every permutation $\pi$, we define a finite set of graphs which are said to {\em expose $\pi$}, and prove the following.
\begin{itemize}
\item(\cref{thm:intchar}) In a hereditary class $\cf C$ of interval graphs, these are equivalent: that (i) $\cf C$ is of bounded twin-width,
	(ii) $\cf C$ excludes all graphs which expose some (fixed) permutation, and (iii) \FO model checking on $\cf C$ is in \FPT (under the assumption~of~ETH).
\item(\cref{thm:intpertu}) In a hereditary class $\cf C$ of graphs which are ``bounded perturbations'' of interval graphs, these are equivalent:
	that (i) $\cf C$ is of bounded twin-width, and (ii) $\cf C$ excludes all graphs which expose some (fixed) permutation.
\end{itemize}

\item In Section~\ref{sec:conclu}, we conclude our findings, state open questions and outline future research directions on the studied topic.
\ifx\proof\inlineproof\else

\item We leave proofs of the ~\apxmark-marked statements for the Appendix.
\fi
\end{itemize}

\medskip
Independently of our research, Bonnet, Chakraborty, Kim, K\"ohler, Lopes and Thomass\'e \cite{DBLP:journals/corr/abs-2204-00722}
have just now come with a result closely related to our \cref{thm:intchar}.
In their terminology, they prove that interval graphs are {\em efficiently delineated} (see \cite[Theorem~23]{DBLP:journals/corr/abs-2204-00722}),
which means that for every hereditary class $\cf C$ of interval graphs has bounded twin-width (with an efficient algorithm to get a contraction sequence),
if and only if the class $\cf C$ does not transduce all finite graphs (see \cref{sub:tranductions} for transductions).
This is equivalent to our \cref{thm:intchar} without the explicit obstructions in part (c)
(although, their proof reveals similar asymptotic obstructions, which is quite natural given the result).
The details of the proofs in \cite{DBLP:journals/corr/abs-2204-00722} and here are not much similar, and
\cite{DBLP:journals/corr/abs-2204-00722} do not include bounded perturbations of interval graphs.

\section{Preliminaries and formal definitions}\label{sec:prel}
%%%%%%%%%%%%%%%%%%%%%%%%%%%%%%%%%%%%%%%%%%%%%%%%%%%%%%%%%%%%%%%%%%%%%%%

\let\bbb\relax
\let\mcc\relax

A (simple) \emph{graph} is a pair $G=(V,E)$ where $V=V(G)$ is the \emph{finite} vertex set and $E=E(G)$ is the edge set -- a set of unordered pairs of vertices $\{u,v\}$, shortly~$uv$.
For a set $Z\subseteq V(G)$, we denote by $G[Z]$ the subgraph of $G$ induced on the vertices of~$Z$.
The neighborhood of a vertex $v\in V(G)$ (implicitly in the graph $G$) is denoted by $N(v)$, 
and two vertices $u,v\in V(G)$ are {\em twins} if $N(u)\setminus\{v\}=N(v)\setminus\{u\}$ (contracting twin vertices, hence, does not create new red edges in the above definition of twin-width).
A graph is {\em twin-free} if it contains no twin pair of vertices.

We will also deal with matrices (as combinatorial objects); a {\em$P\times Q$ matrix} is a matrix whose rows and columns are indexed by the linearly ordered sets $P$ and $Q$, respectively.
A {\em submatrix} of a matrix $\mx A=(a_{i,j}:i\in P,\,j\in Q)$ is, for any $P_1\subseteq P$ and $Q_1\subseteq Q$, the matrix $\mx A'=(a_{i,j}:i\in P_1,\,j\in Q_1)$
which we shortly denote by $\mx A'=\mx A[P_1,Q_1]$.
If $\pi$ is a permutation of a set $M$, then the {\em permutation matrix of $\pi$} is (in an implicit order on $M$) the $M\times M$ matrix $\mx P_{\pi}=(p_{i,j}:i,j\in M)$ such that $p_{i,j}=1$ if $\pi(i)=j$ and $p_{i,j}=0$ otherwise.

\vspace*{-1ex}
\subsection{Intersection graphs}
\vspace*{-1ex}

The \emph{intersection graph} $G$ of a finite collection of sets $\{S_1, \dots, S_n\}$ is a simple graph in which each set
$S_i$ is associated with a vertex $v_i \in V(G)$ (then $S_i$ is the {\em representative} of~$v_i$), 
and each pair $v_i,v_j$ of vertices is joined by an edge if and only if the corresponding sets have a non-empty intersection, i.e.~$v_iv_j \in E(G) \iff S_i \cap S_j \neq \emptyset$.
% We say that an intersection graph $G$ is {\em proper} if $G$ is the intersection graph of $\{S_1, \dots, S_n\}$ such that $S_i\not\subseteq S_j$ for all $i\not=j\in\{1,\ldots,n\}$.

A traditional example of intersection graphs are {\em interval graphs}, which are the intersection graphs of intervals on the real line.
Another example we are interested in are {\em circle graphs} \cite{DBLP:journals/combinatorica/Bouchet87}, which are the intersection graphs of chords of a circle.
A special subcase of circle graphs are {\em permutation graphs}, which are the intersection graphs of chords which join the lower semicircle with the upper semicircle.
The more traditional definition relates a permutation graph $G$ to a permutation $\pi$ on the set $V(G)$, such that $uv\in E(G)$ if and only if $u,v$ is an {\em inversion} in~$\pi$
-- meaning that, in an implicit linear order $\leq$ on $V(G)$, we have $u\leq v\!\iff\! \pi(v)\leq\pi(u)$.
One can easily see equivalence of these definitions in the permutation of the $x$-axis orders of the chord ends in the lower and upper semicircle.

\begin{figure}[t]
$$
\begin{tikzpicture}[scale=0.8]
    \draw[dashed] (0,0) circle[radius=2cm];
    \draw[thick] (-90:2*0.95) --  (-90:2/0.95) node[below] {0};
    \draw[thick] (-130:2) -- (30:2);
    \node (b1) at (-130:2cm+0.7em) {$b_1$};
    \node (b2) at (30:2cm+0.7em) {$b_2$};
    \draw[thick] (150:2) -- (50:2);
    \node (d1) at (150:2cm+0.7em) {$d_1$};
    \node (d2) at (50:2cm+0.7em) {$d_2$};
    \draw[thick] (70:2) -- (-40:2);
    \node (f1) at (70:2cm+0.7em) {$f_1$};
    \node (f2) at (-40:2cm+0.7em) {$f_2$};
    \draw[thick] (-110:2) -- (170:2);
    \node (a1) at (-110:2cm+0.7em) {$a_1$};
    \node (a2) at (170:2cm+0.7em) {$a_2$};
    \draw[thick] (-170:2) -- (110:2);
    \node (c1) at (-170:2cm+0.7em) {$c_1$};
    \node (c2) at (110:2cm+0.7em) {$c_2$};
    \draw[thick] (130:2) -- (-20:2);
    \node (e1) at (130:2cm+0.7em) {$e_1$};
    \node (e2) at (-20:2cm+0.7em) {$e_2$};
    \draw[thick] (0:2) -- (-60:2);
    \node (g1) at (0:2cm+0.7em) {$g_1$};
    \node (g2) at (-60:2cm+0.7em) {$g_2$};
\end{tikzpicture}
\qquad
\begin{tikzpicture}[scale=0.4]
\draw[dashed] (0,0) node[below] {0} --
      (1,0) node[below] {$a_1$} --
      (2,0) node[below] {$b_1$} --
      (5,0) node[below] {$c_1$} --
      (6,0) node[below] {$a_2$} --
      (7,0) node[below] {$d_1$} --
      (8,0) node[below] {$e_1$} --
      (9,0) node[below] {$c_2$} --
      (11,0) node[below] {$f_1$} --
      (12,0) node[below] {$d_2$} --
      (13,0) node[below] {$b_2$} --
      (15,0) node[below] {$g_1$} --
      (16,0) node[below] {$e_2$} --
      (17,0) node[below] {$f_2$} --
      (18,0) node[below] {$g_2$} --
      (20,0) node[below] {$2\pi$};
      \draw[thick] (1,0) to [out=40,in=140] (6,0);
      \draw[thick] (2,0) to [out=70,in=180] (7.5,4.9) to [out=0,in=110](13,0);
      \draw[thick] (5,0) to [out=35,in=145] (9,0);
      \draw[thick] (7,0) to [out=40,in=140] (12,0);
      \draw[thick] (8,0) to [out=60,in=120] (16,0);
      \draw[thick] (11,0) to [out=45,in=135] (17,0);
      \draw[thick] (15,0) to [out=30,in=150] (18,0);
\end{tikzpicture}
\vspace*{-3ex}$$
\caption{``Opening'' a circle representation (left; an intersecting system of chords of a circle)
	into an interval-overlap representation (right; the depicted arcs to be flattened into intervals on the line).
	}
\label{fig:open-circlerep}
\end{figure}
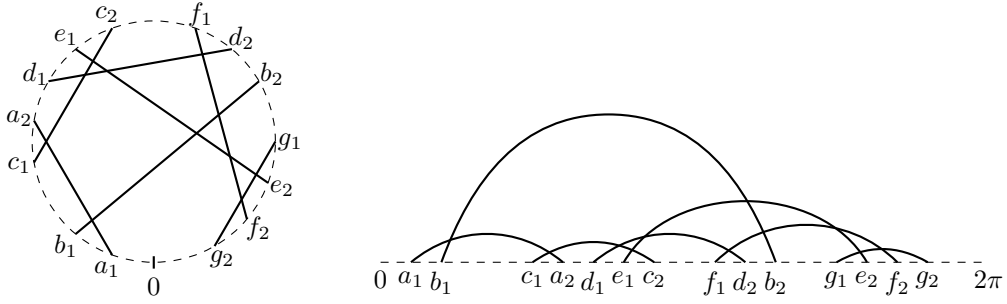

For our paper, it will be useful to note the following alternative definition of circle graphs as {\em interval-overlap graphs}, see \cref{fig:open-circlerep}.
An interval-overlap representation of a graph is a set of intervals on the real line such that two vertices are adjacent, if and only if
their intervals intersect, but one is not strictly contained in the other (strictly with respect to both ends).
The latter strictness condition means that two intervals sharing an endpoint (at any end) {\em do} make an edge.
This condition is usually formulated the other way round, but the presented formulation is convenient with respect to coming \cref{def:ilrepres}.
Of course, it does not change the represented class since we can always perturb the representation to make the ends pairwise distinct in the right direction.
As \cref{fig:open-circlerep} shows, there is a trivial correspondence between circle and interval-overlap representations of the same graph.

\vspace*{-1ex}
\subsection{Twin-width}\label{sub:twwdef}
\vspace*{-1ex}

\begin{figure}[tb]
\begin{tikzpicture}
\begin{scope}[every node/.style={circle,draw,inner sep=1pt}]
\node (a) at (0.5,1) {$a$};
\node (b) at (0,-1) {$b$};
\node (c) at (0,0) {$c$};
\node (d) at (1,0) {$d$};
\node (e) at (1,-1) {$e$};
\path (a) edge (c);
\path (a) edge (d);
\path (b) edge (c);
\path (b) edge (e);
\path (c) edge (d);
\path (d) edge (e);
\end{scope}
\node at (0.5,-2) {};
\node at (-1, 0) {{\boldmath{$G$}}};
\end{tikzpicture}
\hfill
\vbox{\hsize0.66\hsize
\newcommand{\red}{\color{red}$r$}
\newcommand{\mc}[1]{\multicolumn{1}{c|}{#1}}
\newcommand{\mcf}[1]{\multicolumn{1}{|c|}{#1}}
\begin{tabular}{cccccc}
&$a$&$b$&$c$&$d$&$e$\\\cline{2-6}
$a$&\mcf{0}&\mc{0}&\mc{1}&\mc{1}&\mc{0}\\\cline{2-6}
$b$&\mcf{0}&\mc{0}&\mc{1}&\mc{0}&\mc{1}\\\cline{2-6}
$c$&\mcf{1}&\mc{1}&\mc{0}&\mc{1}&\mc{0}\\\cline{2-6}
$d$&\mcf{1}&\mc{0}&\mc{1}&\mc{0}&\mc{1}\\\cline{2-6}
$e$&\mcf{0}&\mc{1}&\mc{0}&\mc{1}&\mc{0}\\\cline{2-6}
\end{tabular}
\hspace{20pt}
\begin{tabular}{ccccc}
&$ab$&$c$&$d$&$e$\\\cline{2-5}
$ab$&\mcf{0}&\mc{1}&\mc{\red}&\mc{\red}\\\cline{2-5}
$c$&\mcf{1}&\mc{0}&\mc{1}&\mc{0}\\\cline{2-5}
$d$&\mcf{\red}&\mc{1}&\mc{0}&\mc{1}\\\cline{2-5}
$e$&\mcf{\red}&\mc{0}&\mc{1}&\mc{0}\\\cline{2-5}
\end{tabular}

\medskip
\smallskip
\begin{tabular}{cccc}
&$ab$&$c$&$de$\\\cline{2-4}
$ab$&\mcf{0}&\mc{1}&\mc{\red}\\\cline{2-4}
$c$&\mcf{1}&\mc{0}&\mc{\red}\\\cline{2-4}
$de$&\mcf{\red}&\mc{\red}&\mc{\red}\\\cline{2-4}
\end{tabular}
\hspace{5pt}
\begin{tabular}{ccc}
&$ab$&$cde$\\\cline{2-3}
$ab$&\mcf{0}&\mc{\red}\\\cline{2-3}
$cde$&\mcf{\red}&\mc{\red}\\\cline{2-3}
\end{tabular}
\hspace{10pt}
\begin{tabular}{|c|}
\hline
\red\\\hline
\end{tabular}
}
\caption{An example of a graph $G$ (left), and of a symmetric contraction sequence of its adjacency matrix (right), which
	corresponds to a contraction sequence in the graph. The maximum red value of $3$ is achieved in the third matrix,
	but the maximum red degree in the corresponding contraction of $G$ is $2$ since the diagonal red entry of the matrix is not reflected in the trigraph.}
        \label{figure:tww}
\end{figure}
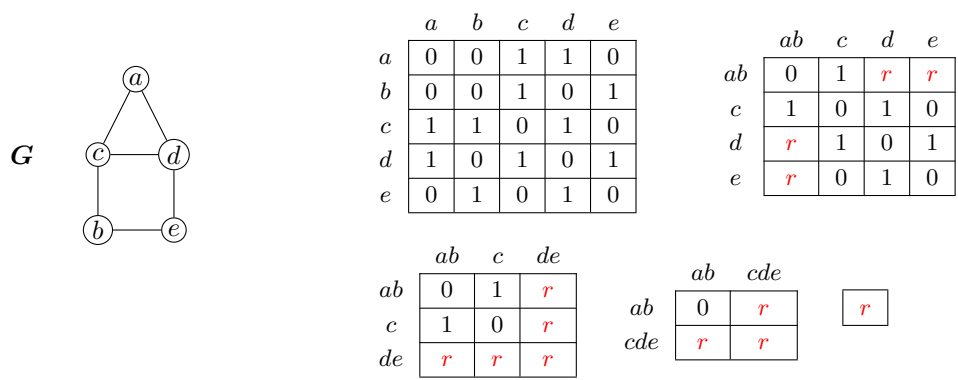

In addition to the previous brief introduction of twin-width through red edges of graphs,
we formally present the definition of twin-width based on matrices, as taken from~\cite[Section~5]{DBLP:conf/focs/Bonnet0TW20}.
The advantage of it is that it applies (through a matrix representation) to arbitrary binary relational structures of finite signature, and gives
a useful asymptotic characterization in \cref{thm:mixedminor}.
See also an illustration in \cref{figure:tww}.

Let $\mx A$ be a matrix with entries from a finite set (e.g., $\{0,1, r\}$ for graphs) and let $R$ and $C$ be the set indexing rows and columns of~$\mx A$.
The special entry $r$ is called a {\em red entry}, and the \textit{red number} of a matrix $\mx A$ is the maximum number of red entries over all columns and rows in $\mx A$.
\emph{Contraction} of two rows (resp. columns) $k$ and $\ell$ results in the matrix obtained by deleting the row (resp. column) $\ell$, 
and replacing entries of the row (resp. column) $k$ by $r$ whenever they differ from the corresponding entries in the row (resp. column)~$\ell$.
A sequence of matrices $\mx A=\mx A_n, \ldots, \mx A_1$ is a \textit{contraction sequence} of $\mx A$, 
if $\mx A_1$ is $(1 \times 1)$ matrix and for all $1 \le i < n$, the matrix $\mx A_i$ is a contraction of matrix $\mx A_{i+1}$. 
The \textit{twin-width} of the matrix $\mx A$ is the minimum integer $d$, such that there exists a contraction sequence $\mx A=\mx A_n, \ldots, \mx A_1$, 
such that for all $1 \le i \le n$, the red number of the matrix $\mx A_i$ is at most $d$.

For example, in case of a graph $G$, we apply the previous to the adjacency matrix of $G$, and we require symmetric contractions,
meaning that the contraction of rows $k,\ell$ is always immediately followed by the contraction of the columns~$k,\ell$.
Then the red entries of the matrices along this contraction sequence correspond to the red edges in trigraphs, except that trigraphs (from the graph definition of twin-width)
do not have red loops corresponding to red entries on the main diagonal.
Hence the graph- and the matrix-based values of twin-width need not be exactly the same, 
but we are in this paper anyway interested in whether the twin-width is asymptotically bounded or not, for which the minor differences are not relevant.

Consider now a fixed linear order $\leq$ on the row and column indices $R$ and $C$ of~$\mx A$.
A~{\em division of $\mx A$} (wrt.~implicit~$\leq$) is a set partition $(R_1,\ldots,R_a)$ of $R$ and $(C_1,\ldots,C_b)$ of $C$ into nonempty consecutive parts in $\leq$
(for any integers $a,b\geq1$), and the corresponding collection of submatrices $\mx A[R_i,C_j]$, $1\leq i\leq a$ and $1\leq j\leq b$, which are called the {\em zones} of this division.
A {\em$k$-mixed minor} in $\mx A$ is a division of $\mx A$ in to $k\times k$ parts (i.e., $a=k=b$ above) such that every zone of it contains two distinct rows and two distinct columns.
The following is a very useful characterization:

\begin{theorem}[Bonnet et al.~\cite{DBLP:conf/focs/Bonnet0TW20}]\label{thm:mixedminor}
a) If a matrix $\mx A$ has twin-width at most $t$, then there is an ordering of the rows and columns of $\mx A$ such that it contains no $(2t+2)$-mixed minor.
\\b) If there is an ordering of the rows and columns of a matrix $\mx A$ such that it contains no $k$-mixed minor, then the twin-width of $\mx A$ is $2^{2^{\ca O(k)}}$,
and the corresponding contraction sequence of $\mx A$ can be computed from the ordering in quadratic time.
\end{theorem}

\subsection{FO logic and transductions}\label{sub:tranductions}
\vspace*{-1ex}

A \emph{relational signature} $\Sigma$ is a finite collection of relational symbols $R_i$, each with associated arity $r_i$. 
A \emph{relational structure} $\cf A$ with signature $\Sigma$ (or shortly a $\Sigma$-structure) is defined by a \emph{domain} $A$ 
and relations $R_i[\cf A] \subseteq A^{r_i}$ for each relation symbol $R_i \in \Sigma$ (the relations \emph{interpret} the relational symbols). 
For example, a graph is a structure with the domain $V$ and a single binary relational symbol $E$.

{\em First-order logic} (abbreviated as \FO) applies to $\Sigma$-structures as follows.
The standard language of first-order logic---including the equality predicate $x=y$,
% an arbitrary number of unary predicates $L(x)$ with the meaning that $x\in A$ holds the label~$L$, 
usual logical connectives $\wedge,\vee,\to$, and quantifiers $\forall x$, $\exists x$ over the domain $A$\,---is
used together with the relational symbols $R_i\in\Sigma$ with the meaning $\cf A\models R_i(\vec x)$ $\iff$ $\vec x\in R_i[\cf A]$.

% Let $\Sigma$ and $\Gamma$ be relational signatures. 
An \emph{\FO interpretation} $\iota$ of $\Gamma$-structures in $\Sigma$-structures is a mapping from $\Sigma$-structures to $\Gamma$-structures 
defined by an \FO formula $\varphi_0(x)$ and an \FO formula $\varphi_R(x_1,\mathellipsis,x_k)$ for each relation symbol $R \in \Gamma$ with arity $k$ 
(these formulas use the relational symbols of $\Sigma$). 
Given a $\Sigma$-structure $\cf A$, $\iota(\cf A)$ is the $\Gamma$-structure whose domain $B$ contains all elements $a \in A$ such that $\cf A\models\varphi_0(a)$, 
and in which every relation symbol $R \in \Gamma$ of arity $k$ is interpreted as the set of tuples $(a_1,\mathellipsis,a_k) \in B^k$ satisfying $\cf A\models\varphi_R(a_1,\mathellipsis,a_k)$.

An {\em\FO transduction $\tau$} (here in a simplified non-copying version), on the other hand, maps from $\Sigma$-structures to subsets of $\Gamma$-structures.
In the first step of $\tau$ -- called {\em parameter expansion}, one maps a $\Sigma$-structure $\cf A$ into a set $\ca A^+$ of $\Sigma^+$-structures,
where $\Sigma^+$ is $\Sigma$ extended with a finite number of unary relation symbols and $\ca A^+$ consists of all possible interpretations of the new symbols.
In the second step, $\tau$ is composed of an \FO formula $\alpha$ (informally, $\alpha$ `marks' valid structures) and an \FO interpretation $\iota$ in $\ca A^+$.
Precisely, the outcome of the transduction is the set $\tau(\cf A)=\{\iota(\cf B): \cf B\in\ca A^+ \wedge\cf B\models\alpha\}$.
For a class of structures $\ca C$, we define $\tau(\ca C)=\bigcup_{\cf A \in \ca C} \tau(\cf A)$.

\section{Common matrix representation of interval-like graphs}\label{sec:ilrepres}
%%%%%%%%%%%%%%%%%%%%%%%%%%%%%%%%%%%%%%%%%%%%%%%%%%%%%%%%%%%%%%%%%%%%%%%

In this section, we study a special matrix representation that can capture in the same way (the ``combinatorial side'' of) both circle and interval intersection representations of graphs.
For this purpose, we consider the interval-overlap view of a circle representation.

\begin{definition}\label{def:ilrepres}\rm
  An {\em interval-like graph representation} -- of a simple graph~$G$ -- is a quadruple $(S,\leq,\eta,\phi)$, where
\begin{itemize}
\item $(S,\leq)$ is a linearly ordered set ($S$ is seen as the set of distinct interval ends),
\item $\eta\subseteq S\times S$ is a binary relation such that $(s_1,s_2)\in\eta$ implies $s_1\leq s_2$ (the pairs $(s_1,s_2)\in\eta$ are seen as the represented intervals, and it may happen that $s_1=s_2$),
  where the pairs of $\eta$ then form the vertex set of $G$ as~$V(G):=\eta$, and
\item  $\phi$ is a predicate of four variables over~$S$ determining whether the vertices represented by $(s_1,s_2)$ and $(t_1,t_2)$ form an edge of~$G$:
  for~$s_1,s_2,t_1,t_2\in S$ such that $(s_1,s_2),(t_1,t_2)\in\eta$ and $(s_1,s_2)\not=(t_1,t_2)$, we have
  $\{(s_1,s_2),\,(t_1,t_2)\}\in E(G)$ iff $\phi(s_1,s_2,t_1,t_2)$ holds true.
%   whether the graph vertices represented by the pairs $(s_1,s_2)$ and $(t_1,t_2)$ form an edge.
\end{itemize}
We say that this is an {\em\FO interval-like representation} if $\phi$ can be expressed as an \FO formula.
\end{definition}

\begin{definition}\label{def:ilmatrix}\rm
  Consider an interval-like graph representation $(S,\leq,\eta,\phi)$, and the set~$R:=\eta\cup\{(t,t):t\in S\}$.
  An {\em il-representation matrix} (`il-' stands for interval-like) is an $R\times S$ matrix $\mx A=\big(a_{i,j}:i\in R,\>j\in S\big)$
  with entries from $\{0,1,2\}$; such that the columns of $\mx A$ are ordered by $\leq$ (of~$S$) and the rows are ordered by the lexicographic power of~$\leq$.
  For $i=(s_1,s_2)\in R$ and~$j\in S$, we have $a_{i,j}=2$ iff $j\lneq s_1$ in $(S,\leq)$, and $a_{i,j}=1$ iff~$i\in\eta$ and $j=s_2$
  (which does not conflict with the previous since $i=(s_1,s_2)\in\eta$ implies~$s_2\geq s_1$).

  We also say that $\mx A$ is an {\em il-representation matrix of the graph~$G$} which is represented by $(S,\leq,\eta,\phi)$ on the vertex set $V(G)=\eta$ as in \cref{def:ilrepres}.
\end{definition}

Both definitions are illustrated in \cref{figure:il-rep}.
\begin{figure}[t]
\hbox{\vbox{\hsize=0.7\hsize
  \begin{tikzpicture}
    \begin{scope}[every node/.style={draw,rectangle,inner sep=0pt,minimum height=4pt,fill=black}]
    \node (ac) at (0,0.2) {};
    \node (bc) at (1,0.4) {};
    \node (be) at (1,0.8) {};
    \node (ca) at (2,0.2) {};
    \node (cb) at (2,0.4) {};
    \node (cd) at (2,0.6) {};
    \node (dc) at (3,0.6) {};
    \node (dd) at (3,0.2) {};
    \node (de) at (3,0.4) {};
    \node (eb) at (4,0.8) {};
    \node (ed) at (4,0.4) {};
    \end{scope}    
    \draw[thick] (ac) -- (ca);
    \draw[thick] (bc) -- (cb);
    \draw[thick] (be) -- (eb);
    \draw[thick] (cd) -- (dc);
    \draw[thick] (de) -- (ed);
    \begin{scope}[every node/.style={draw=none}]
    \draw[dashed] (-0.5,0) -- (4.5,0);
    \node[label=above:$a$] (a) at (0,-0.6) {};
    \node[label=above:$b$] (b) at (1,-0.6) {};
    \node[label=above:$c$] (c) at (2,-0.6) {};
    \node[label=above:$d$] (d) at (3,-0.6) {};
    \node[label=above:$e$] (e) at (4,-0.6) {};
    \end{scope}
  \end{tikzpicture}  
~~
  \begin{tikzpicture}[scale=0.8]
    \begin{scope}[every node/.style={circle,draw,inner sep=1pt}]\scriptsize
    \node (ac) at (0,1) {$ac$};
    \node (bc) at (1,0) {$bc$};
    \node (be) at (1,2) {$be$};
    \node (cd) at (2,1) {$cd$};
    \node (dd) at (3,1) {$dd$};
    \node (de) at (3,2) {$de$};
    \path (ac) edge (bc);
    \path (ac) edge (cd);
    \path (ac) edge (be);
    \path (bc) edge (cd);
    \path (bc) edge (be);
    \path (be) edge (cd);
    \path (be) edge (dd);
    \path (be) edge (de);
    \path (cd) edge (dd);
    \path (cd) edge (de);
    \path (dd) edge (de);
    \end{scope}
  \end{tikzpicture}

\medskip
  \begin{tikzpicture}
    \begin{scope}[every node/.style={draw=none, inner sep=0pt}]
    \draw[dashed] (-0.5,0) -- (4.5,0);
    \node[label=above:$a$] (aa) at (0,-0.4) {};
    \node[label=above:$b$] (bb) at (1,-0.4) {};
    \node[label=above:$c$] (cc) at (2,-0.4) {};
    \node[label=above:$d$] (dd) at (3,-0.4) {};
    \node[label=above:$e$] (ee) at (4,-0.4) {};
    \node (a) at (0,0) {};
    \node (b) at (1,0) {};
    \node (c) at (2,0) {};
    \node (d) at (3,0) {};
    \node (e) at (4,0) {};
    \draw[thick] (a) to [out=40,in=140] (c);
    \draw[thick] (b) to [out=40,in=140] (c);
    \draw[thick] (b) to [out=40,in=140] (e);
    \draw[thick] (c) to [out=40,in=140] (d);
    \draw[thick] (d) to [out=40,in=140] (e);
    \end{scope}
    \draw[dashed] (-0.5,0) -- (4.5,0);
    \node[draw,circle,fill=black,inner sep=1.5pt] (dd) at (3,0) {}; 
    \node[draw,circle,fill=white,inner sep=1pt] (ddd) at (3,0) {}; 
  \end{tikzpicture}
~~
  \begin{tikzpicture}[scale=0.8]
    \begin{scope}[every node/.style={circle,draw,inner sep=1pt}]\scriptsize
    \node (ac) at (0,1) {$ac$};
    \node (bc) at (1,0) {$bc$};
    \node (be) at (1,2) {$be$};
    \node (cd) at (2,1) {$cd$};
    \node (dd) at (3,1) {$dd$};
    \node (de) at (3,2) {$de$};
    \path (ac) edge (bc);
    \path (ac) edge (cd);
    \path (ac) edge (be);
    \path (bc) edge (cd);
    \path (bc) edge (be);
    \path (be) edge (de);
    \path (cd) edge (dd);
    \path (cd) edge (de);
    \path (dd) edge (de);
    \end{scope}
  \end{tikzpicture}
}
\vbox{%\vspace*{-6ex}
\newcommand{\red}{\color{blue} 2}
\newcommand{\mc}[1]{\multicolumn{1}{c|}{#1}}
\newcommand{\mcf}[1]{\multicolumn{1}{|c|}{#1}}
\begin{tabular}{cccccc}
&a&b&c&d&e\\\cline{2-6}
(a,a)&\mcf{0}&\mc{0}&\mc{0}&\mc{0}&\mc{0}\\\cline{2-6}
(a,c)&\mcf{0}&\mc{0}&\mc{\bf1}&\mc{0}&\mc{0}\\\cline{2-6}
(b,b)&\mcf{\red}&\mc{0}&\mc{0}&\mc{0}&\mc{0}\\\cline{2-6}
(b,c)&\mcf{\red}&\mc{0}&\mc{\bf1}&\mc{0}&\mc{0}\\\cline{2-6}
(b,e)&\mcf{\red}&\mc{0}&\mc{0}&\mc{0}&\mc{\bf1}\\\cline{2-6}
(c,c)&\mcf{\red}&\mc{\red}&\mc{0}&\mc{0}&\mc{0}\\\cline{2-6}
(c,d)&\mcf{\red}&\mc{\red}&\mc{0}&\mc{\bf1}&\mc{0}\\\cline{2-6}
(d,d)&\mcf{\red}&\mc{\red}&\mc{\red}&\mc{\bf1}&\mc{0}\\\cline{2-6}
(d,e)&\mcf{\red}&\mc{\red}&\mc{\red}&\mc{0}&\mc{\bf1}\\\cline{2-6}
(e,e)&\mcf{\red}&\mc{\red}&\mc{\red}&\mc{\red}&\mc{0}\\\cline{2-6}
\end{tabular}
\vspace*{-0.5ex}
}}
  
\caption{Illustration of an interval-like representation, with $S=\{a,b,c,d,e\}$ in this order, and $\eta=\{(a,c),\,(b,c),\,(b,e)\,(c,d),\,{(d,d)},\,(d,e)\}$.
(Top-left) the $6$-vertex graph represented by it as an interval graph, (bottom-left) the $6$-vertex graph represented by it as an interval-overlap graph,
and (right) the il-representation matrix (common to both graphs).
The ``dummy'' matrix rows (with no entry~$1$) indexed by $(t,t)\in R\setminus\eta$ help to encode the linear order $(S,\leq)$ of the columns.}
        \label{figure:il-rep}
\end{figure}
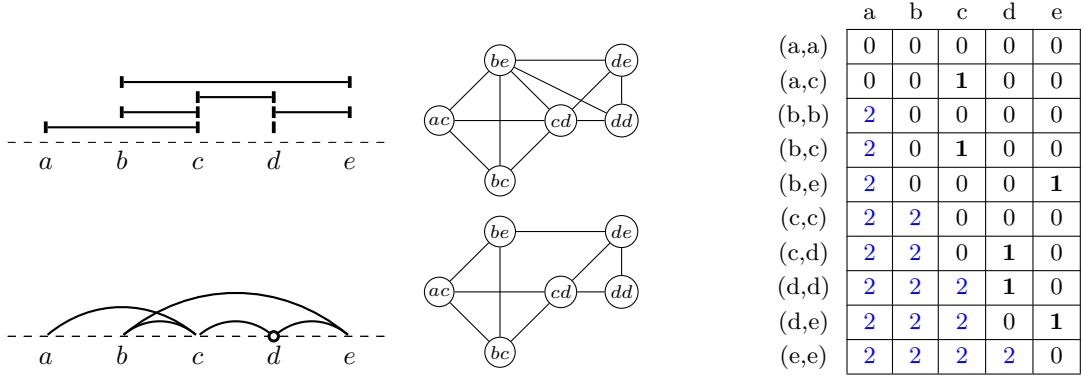

\begin{lemma}\label{lem:haveilrepres}
The classes of interval and of interval-overlap graphs have \FO interval-like representations.
Given an interval or interval-overlap graph $G$, its interval-like representation and the il-representation matrix -- such that $\leq$ is the same as
the order of interval ends in a (respective) intersection representation of~$G$, can be computed in polynomial time.
\end{lemma}
\begin{proof}
There are traditional polynomial time algorithms for the recognition and construction of an intersection representation of interval \cite{recogIntervalLinear}
and interval-overlap \cite{DBLP:journals/jal/Spinrad94} graphs.
Furthermore, in both cases it is straightforward to possibly modify the obtained representation so that it does not contain multiple copies of the same interval.
The representation immediately gives the linearly ordered set $(S,\leq)$ of interval ends, the relation $\eta$ of the intervals, and then the il-representation matrix~$\mx A$.

It remains to give the formula $\phi$ determining the edges of $G$ from the representation (which is, of course, different for each of the two classes).
For interval graphs, it is
$$\phi(s_1,s_2,t_1,t_2)\equiv \alpha(s_1,s_2,t_1,t_2)\vee\alpha(t_1,t_2,s_1,s_2) \mbox{ where }
	\alpha(s_1,s_2,t_1,t_2)\equiv(s_1\leq t_1\leq s_2).$$
For interval-overlap graphs, we simply replace $\alpha$ with $\alpha'$ where
$$\alpha'(s_1,s_2,t_1,t_2)\equiv(s_1\leq t_1\leq s_2\leq t_2). \eqno{\qedhere}$$
\end{proof}

% \begin{definition}\label{def:ilmatrix}\rm
%   Consider an interval-like graph representation $(S,\leq,\eta,\phi)$ as above.
%   An {\em il-representation matrix} (`il-' stands for interval-like) is a square $S\times S$ matrix $\mx A=(a_{i,j}:i,j\in S)$ 
%   with entries from $\{0,1,2\}$ and with the rows and columns ordered by $\leq$; such that
%   $a_{i,j}=2$ iff $j\lneq i$ in $(S,\leq)$, and $a_{i,j}=1$ iff $(i,j)\in\eta$ (which does not conflict with the previous since $(i,j)\in\eta$ implies $i\geq j$).
% %
%   We also say that $\mx A$ is an {\em il-representation matrix of the graph} $G$ which is represented by the considered interval-like representation $(S,\leq,\eta,\phi)$.
% \end{definition}

Regarding FO logic, the matrix $\mx A$ of \cref{def:ilmatrix} is viewed as a relational structure with the domain $R\cup S$ and two binary predicates $\mx A_1$ and $\mx A_2$,
where $\mx A_k(i,j)$ holds for $i,j\in R\cup S$ and $k=1,2$, if and only if $i\in R$, $j\in S$ and $a_{i,j}=k$ in the matrix.
We have:

\begin{lemma}\label{lem:twointerpret}
  There exist \FO interpretations $\iota_i$ and $\iota_o$ such that the following holds.
  For every interval graph (resp., interval-overlap graph) $G$ and its il-repre\-sentation matrix $\mx A$, the interpretation $\iota_i(\mx A)$ (resp., $\iota_o(\mx A)$)
  results in a graph isomorphic~to~$G$.

  Namely, we have $\iota_i=(\sigma_0,\sigma_E)$ and $\iota_o=(\sigma_0,\sigma_E')$ where
\begin{itemize}
\item $\sigma_0(x)\equiv \exists c\, \mx A_1(x,c)$ in both interpretations,
\item $\sigma_E(x,y)$ is the irreflexive and symmetric closure of the formula
	$\forall c\,[\neg\mx A_2(x,c)\vee\mx A_2(y,c)] \>\wedge\> \forall c\,[\mx A_1(x,c)\to\neg\mx A_2(y,c)]$,
% $\big(\neg\mx A_2(s_1,t_1)\wedge\neg\mx A_2(t_1,s_2)\big)$,
\item and similarly $\sigma_E'(x,y)$ is the irreflexive and symmetric closure of~
	$\forall c\,[\neg\mx A_2(x,c)\vee\mx A_2(y,c)] \>\wedge\> \forall c\,[\mx A_1(x,c)\to\neg\mx A_2(y,c)]
		\>\wedge\> \forall c,c'\big[\big(\mx A_1(x,c)\wedge\mx A_1(y,c')\big)\to
			\forall r\big(\mx A_2(r,c)\to\mx A_2(r,c') \big)\big]$.
% $\big(\neg\mx A_2(s_1,t_1)\wedge\neg\mx A_2(t_1,s_2)\wedge\neg\mx A_2(s_2,t_2)\big)$.
\end{itemize}
\end{lemma}
\begin{proof}
  The formula $\sigma_0$ correctly identifies the ground set $V(G)=\eta$ by \cref{def:ilmatrix}.
  Considering $x=(s_1,s_2)\in\eta$ and $y=(t_1,t_2)\in\eta$, we have that $\forall c\,[\neg\mx A_2(x,c)\vee\mx A_2(y,c)]$
  if and only if $s_1\leq t_1$, and $\forall c\,[\mx A_1(x,c)\to\neg\mx A_2(y,c)]$ if and only if $t_1\leq s_2$,
  which both directly follow from \cref{def:ilmatrix}.
  A bit more complicated is the meaning of the last subformula $\forall c,c'\big[\big(\mx A_1(x,c)\wedge\mx A_1(y,c')\big)\to
  \forall r\big(\mx A_2(r,c)\to\mx A_2(r,c') \big)\big]$; this claims that for the columns $c=s_2$ and $c'=t_2$ which are,
  by \cref{def:ilmatrix}, uniquely identified by entries $1$ of $\mx A$ in the rows $x$ and $y$, we have $s_2\leq t_2$ again from \cref{def:ilmatrix}
  (the ``$(t,t)$-rows'' are important).

  The rest follows by the formulas $\alpha$ and $\alpha'$ from the proof of \cref{lem:haveilrepres}.
\end{proof}

One particularly useful consequence of the previous claims is that \FO model checking is in \FPT for graphs which have il-representation matrices of bounded twin-width.
However, one must be careful with right formulation of the opposite direction; since the aforementioned algorithms for interval and interval-overlap representations
do not guarantee to give a representation leading to an ordered matrix without large mixed minors, and one can usually come with il-representation matrices of 
large mixed minors even if those without them exist as well.
We overcome this formulation problem with the following technical definition.

In an interval-like representation $(S,\leq,\eta,\phi)$, we say that we {\em unify} two ends $s_1,s_2\in S$ which are consecutive in $\leq$ (geometrically, we align the two neighboring interval ends to the same point)
if we replace $S$ with $S':=S\setminus\{s_2\}$, and define the relation $\eta'$ as the restriction of $\eta$ on $S_1\setminus\{s_1\}$, while for every $t\in S$ we
define $(s_1,t)\in\eta'$ iff $(s_1,t)\in\eta \vee (s_2,t)\in\eta$, and $(t,s_1)\in\eta'$ iff $(t,s_1)\in\eta \vee (t,s_2)\in\eta$.
Such a unification is {\em legal} if for no $t\in S$ we get $(s_1,t),(s_2,t)\in\eta$ or $(t,s_1),(t,s_2)\in\eta$.
% (notice that problems with legality of a unification especially arise when the represented graph has twin vertices).

An interval-like representation $(S,\leq,\eta,\phi)$ of a graph $G$ is {\em condensed} if there is no legal unification of consecutive ends in $S$
such that the resulting representation $(S',\leq,\eta',\phi)$ would represent a graph isomorphic to~$G$
(the condensed form may not be unique, especially when the represented graph has twin vertices).
An il-representation matrix is {\em condensed} if it comes from a condensed interval-like representation.
% With this terminology we have:
Using \cref{lem:twointerpret} and \cite{DBLP:conf/focs/Bonnet0TW20}, we obtain:

\begin{lemma2rep}\apxmark\label{lem:twwtomc}
Assume that there is an integer constant $b$, such that {\em every} condensed il-representation matrix of an interval or interval-overlap graph $G$
has no $b$-mixed minor.
Then the \FO model checking problem of $G$ can be solved in \FPT time with respect to the formula size.
\end{lemma2rep}
\begin{proof}
We compute an arbitrary interval-like representation of $G$ using \cref{lem:haveilrepres}.
Then we greedily make it condensed, which is trivial in polynomial time, and construct the il-representation matrix $\mx A$ of $G$ by \cref{def:ilmatrix}.
Then $\mx A$ is of bounded twin-width by \Cref{thm:mixedminor} and, moreover, by \cite[Section~5]{DBLP:conf/focs/Bonnet0TW20} we can
in polynomial time construct a contraction sequence for $\mx A$ of bounded red degree.
Then, by \cite[Section~7]{DBLP:conf/focs/Bonnet0TW20}, we solve \FO model checking on $\mx A$ in \FPT time with respect to $b$ and the checked formula size.

Now we turn the attention to the \FO model checking problem of our graph $G$.
Given an interval or interval-overlap graph $G$ and a formula $\varphi$ on $G$, we take the respective interpretation $\iota\in\{\iota_i,\iota_o\}$ of \cref{lem:twointerpret},
and syntactically translate $\varphi$ into an \FO formula $\iota(\varphi)$ such that $(G\simeq \iota(\mx A))\models\varphi$ $\iff$ $\mx A\models \iota(\varphi)$.
Then we call the \FPT model-checking algorithm on~$\mx A$.
For the sake of completeness, we outline the standard translation of $\varphi$ to $\iota(\varphi)$ where $\iota=(\sigma_0,\sigma_E)$;
every occurrence of a quantifier $\exists v(\alpha)$ in $\varphi$ is replaced with $\exists v\big(\sigma_0(v)\wedge\alpha\big)$,
and every occurrence of the edge predicate $\prebox{edge}(v,w)$ is replaced with $\sigma_E(v,w)$.
\end{proof}

In the opposite situation to \cref{lem:twwtomc}, when some il-representation matrix $\mx A$ of our graph $G$ contains a large (unbounded in the asymptotic setting) mixed minor,
we will ``extract'' from $\mx A$ a substructure which will, in the coming sections, serve as an explicit certificate of unbounded twin-width of~$G$.
This may seem contradictory at the first sight, as mentioned above, since we may come up with an il-representation matrix $\mx A$ of $G$ of large mixed minor, 
and still may possibly have another il-representation matrix of $G$ without such, which would certify bounded twin-width of~$G$ via \cref{lem:twointerpret} and \cite{DBLP:conf/focs/Bonnet0TW20}.
However, an additional assumption of a condensed il-representation matrix will make this approach sound.
We will hence (later) implicitly show that the twin-width of various condensed il-representation matrices of the same graphs cannot ``range from bounded to unbounded''.

% \ifx\proof\inlineproof\else
Before we get to the specific obstructions, we derive from \Cref{thm:mixedminor} the following tool:
% \fi
\begin{lemma2rep}\apxmark\label{lem:getpermm}
% For every integer $p$ there exists $q$ such that the following holds. 
Let $G$ be an interval or interval-overlap graph.
If we have an il-representation matrix $\mx A$ of $G$ with a $(2p+1)$-mixed minor, then, for every permutation $\pi$ of $p$ elements,
there exists a $p\times p$ submatrix $\mx A'$ of $\mx A$ equal to $\mx P_{\pi}$, such that
for every two distinct rows of $\mx A'$ indexed by $(s_1,s_2)$ and $(t_1,t_2)$ (cf.~\cref{def:ilmatrix}) we have $s_1\not=t_1$.
% and $\mx A'=\mx P_{\pi}$ -- the permutation matrix of~$\pi$.
\end{lemma2rep}

\begin{proof}
% We choose $q$ accordingly to \Cref{thm:mixedminor}; such that any $\{0,1,2\}$-matrix of twin-width at least $p$ contains
% a $(2p+1)\times(2p+1)$ mixed minor.
% Then we apply this to the assumed $R\times S$ il-representation matrix $\mx A$ of $G$ of twin-width at least~$q$, and thus obtain
In the assumed $R\times S$ il-representation matrix $\mx A$ of $G$ (cf.~\cref{def:ilmatrix}), we pick
a $(2p+1)$-division $D=(\ca P,\ca Q)$ of $\mx A$ such that each zone is mixed -- in particular, each zone of the division $D$ contains a nonzero entry.
In this division, $\ca P=(P_1,\ldots,P_{2p},P_{2p+1})$ is a consecutive partition of the row index set $R$, 
and $\ca Q=(Q_0,Q_1,\ldots,Q_{2p})$ is a consecutive partition of the column index set $S$.
Observe that if any zone which is not among the lowest ones or among the leftmost ones, contained an entry $2$, then
the lowest leftmost zone $\mx A[P_{2p+1},Q_0]$ in $\mx A$ would be homogeneous of all entries $2$ by \cref{def:ilmatrix}, a contradiction to it being mixed.
Hence all zones $\mx A[P_i,P_j]$ such that $1\leq i,j\leq2p$ must contain only entries $0$ or $1$.

Assume $\pi$ to be a permutation of $\{1,\ldots,p\}$.
Then every zone $\mx A[P_{2k},Q_{2\pi(k)}]$, $k=1,\ldots,p$, contains a nonzero entry $a_{r_k,c_k}=1$,
and we choose the submatrix $\mx A':=\big(a_{i,j}:i\in\{r_1,\ldots,r_p\},j\in\{c_1,\ldots,c_p\}\big)$.
We have that every row $k$ of $\mx A'$ contains only one entry $1$ by \cref{def:ilmatrix}, and so it is our $a_{r_k,c_k}=1$
and, indeed, $\mx A'=\mx P_{\pi}$.
It remains to verify the claimed ``index condition'' on the rows of $\mx A'$.
Assume the contrary, that for rows $r_k=(s_1,s_2)$ and $r_{k'}=(t_1,t_2)$ of $\mx A'$ we have $k<k'$ and $s_1=s_2$.
However, then by the lexicographic ordering of row indices of $\mx A$ from \cref{def:ilmatrix},
we get that the zone $\mx A[P_{2k'-1},Q_{2\ell-1}]$ where $\ell=\min(\pi(k),\pi(k'))$ cannot contain any nonzero entry
again by \cref{def:ilmatrix} (the row indices of this zone are actually of the form $(s_1,t)$ where $\min(c_k,c_{k'})<t<\max(c_k,c_{k'})$).
This contradicts the fact that $\mx A[P_{2k'-1},Q_{2\ell-1}]$ is mixed.
The statement is proved.
\end{proof}

\section{Circle graphs}\label{sec:circle}
%%%%%%%%%%%%%%%%%%%%%%%%%%%%%%%%%%%%%%%%%%%%%%%%%%%%%%%%%%%%%%%%%%%%%%%

We now formulate and prove our first main result.
\begin{theorem}\label{thm:circlechar}
Let $\cf C$ be a hereditary class of circle graphs, and let $\cf A$ be the class of all possible il-representation matrices of graphs from $\cf C$.
Then the following are equivalent:
\begin{enumerate}[a)]
\item $\cf C$ is of bounded twin-width,
\item there is an integer $b$ such that no ordered matrix of $\cf A$ contains a $b$-mixed minor,
\item $\cf C$ does not contain all permutation graphs,
\item \FO model checking on $\cf C$ is in \FPT (under the assumption \FPT$\!\not=\,${\sf AW[*]}, or usual ETH).
\end{enumerate}
\end{theorem}

\begin{proof}
Recall that every permutation graph is also a circle graph, and that the class of circle graphs is the same as the class of interval-overlap graphs
(with a direct translation between the representations, cf.~\cref{fig:open-circlerep}).
We first show the equivalence of (a), (b) and (c).

If $\cf C$ contains all permutation graphs, then $\cf C$ is of unbounded twin-width 
since the class of all permutation graphs has unbounded twin-width~\cite{DBLP:conf/focs/Bonnet0TW20}, establishing (a)$\Rightarrow$(c).
For (b)$\Rightarrow$(a), we start with \Cref{thm:mixedminor}, deriving that the unordered matrix family $\cf A$ is of bounded twin-width.
Then we use that the property of bounded twin-width is preserved under \FO interpretations 
and more generally transductions of relational structures by~\cite[Section~8]{DBLP:conf/focs/Bonnet0TW20}.
Hence, if $\cf A$ is of bounded twin-width, then so is $\cf C$ by \cref{lem:twointerpret}.

We are left with the implication (c)$\Rightarrow$(b), for which we assume that no such bound $b$ exists and choose any permutation $\pi$ of a $p$-element set for some~$p$.
Thanks to our assumption, we can pick an $R\times S$ matrix $\mx A\in\cf A$ (an il-representation of a graph $G\in\cf C$)
such that $\mx A$ contains a $p\times p$ submatrix $\mx A'=\mx P_{\pi^{-1}}$ as claimed by \cref{lem:getpermm} (note the inverse permutation~$\pi^{-1}$).
Let $P'\subseteq R$ be the row subset of $\mx A'$ and $Q'\subseteq S$ be the column subset of it.
Since every row of $\mx A'$ has an entry $1$, we actually get $P'\subseteq V(G)$.

For any two distinct vertices $x,y\in P'$ giving entries $a_{i,j}=1$ and $a_{i',j'}=1$ in $\mx A'$, we have $i=(s,j)$ and $i'=(s',j')$ where $s,s'\in S$ by \cref{def:ilmatrix},
and $s\leq s'$ up to symmetry. Furthermore, $j\geq s'$ since $\mx A'$ has no entry $2$.
By the definition of an interval-overlap graph, we the get that $xy\in E(G)$ if and only if $j\leq j'$, which means that the pair $i,i'$ is {\em not} an inversion in~$\pi^{-1}$.
In other words, $xy\in E(G)$ iff $i$ and $i'$ is an inversion of $\pi$, and the induced subgraph $G[P']$ is the permutation graph of~$\pi$.
Since $\cf C$ is hereditary, it thus contains all permutation graphs.

Lastly, we look at (d). The implication (b)$\Rightarrow$(d) has been proved in \cref{lem:twwtomc}.
Conversely, (d)$\Rightarrow$(c) follows since there exists an \FO transduction from the class of all permutation graphs
to the class of all graphs \cite{DBLP:journals/comgeo/HlinenyPR19,DBLP:journals/corr/abs-2102-06880}, and \FO model checking
on the class of all graphs is {\sf AW[*]}-complete \cite{DBLP:conf/dmtcs/DowneyFT96}.
The fine assumption \FPT$\!\not=\,${\sf AW[*]} is, moreover, implied by the more usual complexity assumption of the Exponential Time Hypothesis (ETH).
\end{proof}

\Cref{thm:circlechar} has a notable extension in \Cref{thm:circlepertu} (unfortunately not yet satisfactory in algorithmic sense).
In order to formulate it, we need to define perturbations of graphs.
An {\em elementary perturbation} of a (simple) graph $G$ is the operation which chooses an arbitrary subset of vertices $X\subseteq V(G)$,
and then complements the edge relation on the pairs from $X$.
An {\em$r$-bounded perturbation of $G$} performs a sequence of at most $r$ elementary perturbations~on~$G$.

The concept of bounded perturbations is not widespread in graph theory, however, it is closely related to low-rank perturbations of matroids
in the matroid-minor structure theory.
In fact, bounded perturbations of circle graphs seem to play a very important role in the ongoing investigation of the structural theory
of vertex-minors of graphs, see~\cite{McCartyRose2021} (in this regard, \cref{thm:circlepertu} thus may find applications in studying
the twin-width of proper vertex-minor closed classes of graphs).

\begin{toappendix}
We make use of the following easy technical claim:
\begin{lemma}\apxmark\label{lem:lexico1}
Let $r\in\mathbb N$ and $s=2^r$. Consider a finite set $Y$ and take the Cartesian power $Z:=Y^s$.
For any collection of sets $X_1,\ldots,X_r\subseteq Z$, there exist $1\leq\ell\leq s$, and $a_i\in Y$ where $i=1,\ldots,\ell-1$,
and $b_j^k\in Y$ where $j=\ell+1,\ldots,s$ and $k\in Y$, such that the following holds:
For the set $Z_0=\{(a_1,\ldots,a_{\ell-1}, \,y,\, b_{\ell+1}^y,\ldots,b_s^y): y\in Y\}\subseteq Z$ and every $i=1,\ldots,r$,
we have $Z_0\subseteq X_i$ or $Z_0\cap X_i=\emptyset$.
\end{lemma}

In the setting of \cref{lem:lexico1}, we call $Z_0$ a {\em homogeneous set with prefix} $(a_1,\ldots,a_{\ell-1})$.
\begin{proof}
We color each element $z\in Z=Y^s$ by one of at most $s$ colors based on the sets among $X_1,\ldots,X_r$ that $z$ belongs to (we thus aim for monochromatic~$Z_0$).
Then we forget $X_1,\ldots,X_r$ and prove the statement by induction on~$s$, where the base case of $s=1$ is trivial.

For $s>1$, we consider the sets $A_x=\{(x,\, c_2,\ldots,c_s): c_2,\ldots,c_s\in Y\}$ where $x\in Y$.
If each of the sets $A_x$ contains an element of color (say) $1$, then we are done by setting $\ell=1$.
Otherwise, some set $A_x$ for $x\in Y$ does not contain any element of color $1$, and we apply induction with $s-1$ colors and the set $A_x$ in place of $Z$. 
This way we get a homogeneous set $Z_1$ with prefix $(a_2,\ldots,a_{\ell-1})$, and we add~$a_1:=x$ to this prefix in order to finish the claim.
\end{proof}
\end{toappendix}

\begin{theorem2rep}\apxmark\label{thm:circlepertu}
Let $r\in\mathbb N$ and $\cf C$ be a hereditary class of graphs, such that every graph in $\cf C$ is an $r$-bounded perturbation of a circle graph.
Then $\cf C$ is of bounded twin-width, if and only if $\cf C$ does not contain all permutation graphs.
\end{theorem2rep}

\begin{proof}
If $\cf C$ contains all permutation graphs (which themselves have unbounded twin-width, as shown already in \cref{thm:circlechar}), then $\cf C$ is of unbounded twin-width.

Conversely, let $\cf C_0$ be a hereditary class of circle graphs, such that every graph in $\cf C$ is 
an $r$-bounded perturbation of a graph in~$\cf C_0$, and that some $r$-bounded perturbation of every graph in $\cf C_0$ falls into~$\cf C$.
% In the converse direction, let $\cf C_0$ be a hereditary subclass of all circle graphs, such that every graph in $\cf C$ is an $r$-bounded perturbation of a graph in~$\cf C_0$.
Since an $r$-bounded perturbation, for constant $r$, can be straightforwardly (although a bit technically) expressed as an \FO transduction, 
bounded twin-width of $\cf C_0$ would imply bounded twin-width of $\cf C$ by \cite[Section~8]{DBLP:conf/focs/Bonnet0TW20}.
Therefore, if $\cf C$ is of unbounded twin-width, then $\cf C_0$ contains all permutation graphs by \cref{thm:circlechar}.
To finish that $\cf C$ contains every permutation graph $H$, too, we construct from $H$ a suitable (much larger) permutation graph $H^+\in\cf C_0$
such that every $r$-bounded perturbation of $H^+$ contains~$H$.
% , we are going to show that if $\cf C_0$ contains all permutation graphs, then so does $\cf C$.

Let $H$ be a $p$-vertex permutation graph determined by a permutation $\pi_1$ of $\{1,\ldots,p\}$,
and let $\pi_2$ be the permutation obtained by concatenation (as an order) of $\pi_1$ followed by the inverse of $\pi_1$ on $\{p+1,\ldots,2p\}$.
Hence, the permutation graph $H_2$ of $\pi_2$, on $V(H_2)=\{1,\ldots,2p\}$, is the disjoint union of $H$ and of the complement of~$H$.
The permutation $\pi_2$ is determined by two linear orders; $1\leq_12\leq_1\ldots\leq_12p$ and $\pi_2(1)\leq_2\ldots\leq_2\pi_2(2p)$.
Let $Z:=\{1,\ldots,2p\}^s$ where $s=2^r$, and let the linear orders $\leq_1$ and $\leq_2$ on $Z$ be 
the standard lexicographic powers of $\leq_1$ and $\leq_2$ on $\{1,\ldots,2p\}$.
Then $\leq_1$ and $\leq_2$ on $Z$ determine a permutation $\varrho$.
% (on $(2p)^s$ elements).

Let $H^+$ be the permutation graph of $\varrho$ on the vertex set $V(H^+)=Z=V(H_2)^{s}$, and consider a graph $G\in\cf C_0$ such that~$H^+\subseteq G$ (up to isomorphism).
Let $G'\in\cf C$ and $X_1,\ldots,X_r\subseteq V(G)$ be the sets to which an $r$-bounded perturbation of $G$ has been applied, restricted to $V(H^+)$, 
such that $G'$ is the perturbed graph obtained from~$G$.
We apply \cref{lem:lexico1} to $Y:=V(H_2)$ (and~$Z=V(H^+)$ as chosen above); this gives us a homogeneous set $Z_0\subseteq Z$ with some prefix $(a_1,\ldots,a_{\ell-1})$.
Since the linear orders $\leq_1$ and $\leq_2$ on $Z$ have been obtained as lexicographic powers of $\leq_1$ and $\leq_2$ on $Y$,
their restriction to $Z_0$ is isomorphic to $\leq_1$ and $\leq_2$ on $Y$.
In particular, the induced subgraph $G[Z_0]$ is isomorphic to $H_2$.
Moreover, by \cref{lem:lexico1}, the complementations of the $r$-bounded perturbation leading to~$G'$ have been applied 
either to none or to all vertices of $Z_0$, and so $G'[Z_0]$ is either $G[Z_0]$ or its complement.
In any case, $G'[Z_0]$ contains an induced subgraph isomorphic to $H$, and so does~$G'\in\cf C$.

Since $H$ has been chosen as an arbitrary permutation graph, and $\cf C$ is hereditary, we get that $\cf C$ contains all permutation graphs.
\end{proof}

\ifx\proof\inlineproof\else
\begin{proof}[Proof sketch]
If $\cf C$ contains all permutation graphs (which themselves have unbounded twin-width, as shown already in \cref{thm:circlechar}), then $\cf C$ is of unbounded twin-width.

Conversely, let $\cf C_0$ be a hereditary class of circle graphs, such that every graph in $\cf C$ is 
an $r$-bounded perturbation of a graph in~$\cf C_0$, and that some $r$-bounded perturbation of every graph in $\cf C_0$ falls into~$\cf C$.
% In the converse direction, let $\cf C_0$ be a hereditary subclass of all circle graphs, such that every graph in $\cf C$ is an $r$-bounded perturbation of a graph in~$\cf C_0$.
Since an $r$-bounded perturbation, for a constant $r$, can be expressed as an \FO transduction, bounded twin-width of $\cf C_0$
would imply bounded twin-width of $\cf C$ by \cite[Section~8]{DBLP:conf/focs/Bonnet0TW20}.
Therefore, if $\cf C$ is of unbounded twin-width, then $\cf C_0$ contains all permutation graphs by \cref{thm:circlechar}.
To finish that $\cf C$ contains every permutation graph $H$, too, we construct from $H$ a suitable (much larger) permutation graph $H^+\in\cf C_0$
such that every $r$-bounded perturbation of $H^+$ contains~$H$.
Details of this construction are left for the appendix.
\end{proof}
\fi

\section{Interval graphs}\label{sec:interval}
%%%%%%%%%%%%%%%%%%%%%%%%%%%%%%%%%%%%%%%%%%%%%%%%%%%%%%%%%%%%%%%%%%%%%%%

Our second main result concerns tractability of \FO model checking on interval graphs,
and it is a bit more complicated to formulate since its ``excluded-something'' condition deals with finite collections of induced subgraphs
to be excluded instead of individual excluded permutation graphs.

Consider a permutation $\pi$ of $\{1,\ldots,p\}$, and disjoint sets $W=\{w_1,\ldots,w_p\}$ and $W_1,W_2$ where~$|W_1|=|W_2|=p=|W|$.
% which is combinatorially determined by two linear orders $1\leq_12\leq_1\ldots\leq_1p$ and $\pi(1)\leq_2\ldots\leq_2\pi(p)$.
Let $H$ be a graph on $W\cup W_1\cup W_2$ such that every vertex $w_i\in W$ has nonempty neighborhoods $N_1(w_i)$ in $W_1$ and $N_2(w_i)$ in $W_2$.
Moreover, $H$ is such that $N_1(w_1)\subsetneq N_1(w_2)\subsetneq\ldots\subsetneq N_1(w_p)$ and $N_2(w_{\pi(1)})\subsetneq N_2(w_{\pi(2)})\subsetneq\ldots\subsetneq N_2(w_{\pi(p)})$.
Edges within each of $W$, $W_1$ and $W_2$, and between $W_1$ and $W_2$ can be arbitrary in general.
We then say that the graph $H$ {\em exposes the permutation} $\pi$.
It is easy to observe that for every permutation $\pi$, there exists an interval graph exposing $\pi$, in particular one in which each of the
sets $W,W_1,W_2$ induces a clique and there is no edge from $W_1$ to $W_2$ -- its interval representation can be simply constructed by 
arranging the left and right ends of the intervals of $W_1$ according to the permutation~$\pi$.

Observe, moreover, that for a graph $H$ which exposes $\pi$, each of the set $N_1(w_{i})\setminus N_1(w_{i-1})$, $2\leq i\leq p$, has to consist of 
a single element of $W_1$ which we for reference call the {\em mate of $w_i$ in $W_1$}. The mate of $w_1$ is the single vertex of $N_1(w_1)$.
We analogously define the {\em mate of $w_i$ in $W_2$} as the singleton in $N_2(w_{\pi(i)})\setminus N_2(w_{\pi(i-1)})$ or $N_2(w_{\pi(1)})$.
Altogether, we call the set $W_1\cup W_2$ the {\em mates of~$W$}.

With respect to the previous definition, we remark that each of the induced subgraphs $H[W\cup W_1]$ and $H[W\cup W_2]$ is tightly related
to \FO definable linear orders in graphs, and such subgraphs are sometimes named ``ladders'' in papers on \FO logic of graphs
(however, ``ladders'' often additionally require that the subgraphs induced on $W$ and $W_1$ ($W_2$) are cliques or independent, which is not convenient for us).

Let $\ca R^+_{\pi}$ denote the (finite) collection of all non-isomorphic graphs $H$ which expose the particular permutation $\pi$,
and let $\ca R_{\pi}\subseteq\ca R^+_{\pi}$ be the subset of those which are interval graphs.
Before the main result, we need a technical claim:
\begin{lemma2rep}\apxmark\label{lem:transduceperm}
There exists an \FO transduction $\tau$ from graphs to permutations such that, for every permutation $\pi$, the following holds.
If $G$ is a graph containing some member $H\in\ca R^+_{\pi}$ as an induced subgraph, then the transduction image $\tau(G)$ contains the permutation~$\pi$.
\end{lemma2rep}
\begin{proof}
In the parameter-expansion part of the desired transduction $\tau$, we ``guess'' three marks $m,m_1,m_2$ 
such that $m$ is intended to represent in $V(G)$ the set $W$ and $m_i$, $i=1,2$, the set $W_i$ of the graph $H$ (cf.~the definition of $H$ exposing~$\pi$).
In the subsequent \FO interpretation, we restrict the domain (of $\pi$) with $\varphi_0(x)\equiv m(x)$, and \FO interpret two intended linear orders
$\leq_1$ and $\leq_2$ on this domain, for $j=1,2$, as $\leq_j\!(x,y)\equiv \forall z\big(m_j(z)\to (\prebox{edge}(x,z)\to\prebox{edge}(y,z))\big)$
(expressing that $N_j(x)\subseteq N_j(y)$).
We can also routinely express in \FO logic that $\leq_j$, $j=1,2$, is a linear order; 
$\alpha_j\equiv \forall x,y\big[(x\leq_jy\vee y\leq_jx)\wedge \big((x\leq_jy\wedge y\leq_jx)\to x=y\big)\big] \wedge
 \forall x,y,z\big[(x\leq_jy\,\wedge$ $\!y\leq_jz)\to x\leq_jz\big]$, and so $\tau(G)$ consists of permutations (on subsets of~$V(G)$).

Now, if the marks $m,m_2,m_3$ exactly coincide with the partition $V(H)=W\cup W_1\cup W_2$ in whole~$G$, then the interpreted
relations $\leq_1$ and $\leq_2$ are indeed linear orders defining $\pi$ as exposed by $H$, and so~$\pi\in\tau(G)$.
\end{proof}

\begin{theorem}\label{thm:intchar}
Let $\cf C$ be a hereditary class of interval graphs, and let $\cf A$ be the class of all {\em condensed} il-representation matrices 
of the {\em twin-free} graphs from $\cf C$. Then the following are equivalent:
\begin{enumerate}[a)]
\item $\cf C$ is of bounded twin-width,
% \item $\cf A$ is of bounded twin-width,
\item there is an integer $b$ such that no ordered matrix of $\cf A$ contains a $b$-mixed minor,
\item for some permutation $\pi$, the class $\cf C$ excludes all graphs in the collection $\ca R_{\pi}$ exposing~$\pi$,
\item \FO model checking on $\cf C$ is in \FPT (under the assumption \FPT$\!\not=\,${\sf AW[*]}, or usual ETH).
\end{enumerate}
\end{theorem}

\begin{proof}
We proceed along the same steps as in the proof of \cref{thm:circlechar}, first showing the equivalence of (a), (b) and (c).
Let $\cf C'$ be the subclass of twin-free members of $\cf C$.

If, for every permutation $\pi$, the class $\cf C$ contains a graph isomorphic to a member of $\ca R_{\pi}$, 
then the transduced class $\tau(\cf C)$ by \cref{lem:transduceperm} contains all permutations.
Then $\tau(\cf C)$ is of unbounded twin-width by \cite{DBLP:journals/corr/abs-2102-06880}, and since the property of bounded twin-width
is preserved under \FO transductions of relational structures by~\cite[Section~8]{DBLP:conf/focs/Bonnet0TW20}, 
we get (a)$\Rightarrow$(c); that $\cf C$ is of unbounded twin-width.
We similarly (as in \ref{thm:circlechar}) deduce (b)$\Rightarrow$(a); from \Cref{thm:mixedminor} we get that $\cf A$ as a family of unordered matrices is of bounded twin-width,
and then so is $\cf C'$ which is \FO interpreted in $\cf A$ by \cref{lem:twointerpret}, 
and each graph in $\cf C\setminus\cf C'$ is obtained from a graph in $\cf C$ by adding twin vertices which does not raise the twin-width.

We are left with the implication (c)$\Rightarrow$(b), for which we assume that no such bound $b$ exists and choose any permutation $\pi$ of a $p$-element set for some~$p$.
Thanks to our assumption, we can pick an $R\times S$ matrix $\mx A\in\cf A$ (an il-representation of a graph $G\in\cf C$)
such that $\mx A$ contains a $p\times p$ submatrix $\mx A'=\mx P_{\pi}$ as claimed by \cref{lem:getpermm}.
Let $P'\subseteq R$ be the row subset of $\mx A'$ and $Q'\subseteq S$ be the column subset of it.
Since every row of $\mx A'$ has an entry $1$, we actually get $P'\subseteq V(G)$.

% We are left with the implication (c)$\Rightarrow$(b), for which we assume that $\cf A$ is of unbounded twin-width,
% and pick any permutation $\pi$ of a $p$-element set for some~$p$.
% Since the twin-width of $\cf A$ is unbounded, we can pick an $S\times S$ matrix $\mx A\in\cf A$ (of a twin-free graph $G\in\cf C'$)
% such that $\mx A$ contains a $p\times p$ submatrix $\mx A'$ as claimed by \cref{lem:getpermm} for our permutation~$\pi$.
% Let $P'\subseteq S$ be the row subset of $\mx A'$ and $Q'\subseteq S$ be the column subset of it.
% Let $X\subseteq V(G)$ be the subset of vertices given by the entries~$1$ in $\mx A'$ (cf.~\cref{def:ilrepres}) which coincide with the $1$'s 
% in the permutation matrix~$\mx P_{\pi}=(p_{i,j}:i\in P',j\in Q')$.

Now comes the main difference from the proof of \cref{thm:circlechar}; we have and will use that $\mx A\in\cf A$ is 
a {\em condensed} il-representation matrix of twin-free~$G$ (the proof would not go through without this assumption).
Pick a vertex $x\in V(G)$ with the entry $a_{i,j}=1$ in $\mx A$, and so $i=(s,j)$ for some $s\in S$;
then $s$ will be called the {\em left end} and $j\in S$ the {\em right end} of~$x$ (which naturally corresponds to the geometric image of an interval representation of $G$ described by~$\mx A$).
We first observe that since $\mx A'$ has no entry $2$, by \cref{def:ilmatrix} all left ends of vertices in $P'$ are to the left (or equal) of all right ends of these vertices.

Since $\mx A$ is condensed (and it represents a twin-free~graph), the left end of $x$ must also be the right end of some other vertex $y$ of $G$ (but $y\not\in P'$), 
as otherwise we could unify this left end with the next higher one in $S$ without changing the represented interval graph.
For reference, we denote by $l(x):=y$. 
Symmetrically, the right end of $x$ must also be the left end of some other vertex $r(x)=z$ of $G$ (where~$z\not\in P'$).
Let $H$ be the subgraph of $G$ induced by the set $X\cup\{l(x),r(x): x\in X\}$.
Then, as one may routinely verify from the definition of an interval representation by $\mx A$, we have~$H\in\ca R_{\pi}$.
Since $\cf C$ is hereditary, it thus contains a graph from $\ca R_{\pi}$ (exposing $\pi$) for every permutation~$\pi$.

Lastly, we look at (d). The implication (b)$\Rightarrow$(d) has again been proved in \cref{lem:twwtomc}.
Conversely, (d)$\Rightarrow$(c) follows since we have got the transduced class $\tau(\cf C)$ of all permutations 
and there exists an \FO transduction from the class of all permutations to the class of all graphs \cite{DBLP:journals/corr/abs-2102-06880}.
As before, \FO model checking on the class of all graphs is {\sf AW[*]}-complete \cite{DBLP:conf/dmtcs/DowneyFT96},
and the fine assumption \FPT$\!\not=\,${\sf AW[*]} is implied by the ETH.
\end{proof}

We can again extend the non-algorithmic part of \cref{thm:intchar} to bounded perturbations of interval graphs,
using an approach similar to the proof of \cref{thm:circlepertu}.

\begin{theorem2rep}\apxmark\label{thm:intpertu}
Let $r\in\mathbb N$ and $\cf C$ be a hereditary class of graphs, such that every graph in $\cf C$ is an $r$-bounded perturbation of an interval graph.
Then $\cf C$ is of bounded twin-width, if and only if, for some permutation $\pi$, the class $\cf C$ excludes all graphs in the collection $\ca R^+_{\pi}$ exposing~$\pi$.
\end{theorem2rep}

\begin{proof}
If $\cf C$ contains some graph from $\ca R^+_{\pi}$ for every permutation $\pi$, then, by \cref{lem:transduceperm},
the transduced class $\tau(\cf C)$ contains all permutations and hence is of unbounded twin-width, exactly as in the proof of \cref{thm:intchar}.

Conversely, let $\cf C_0$ be a hereditary class of interval graphs, such that every graph in $\cf C$ is 
an $r$-bounded perturbation of a graph in~$\cf C_0$, and that some $r$-bounded perturbation of every graph in $\cf C_0$ falls into~$\cf C$.
Again, since an $r$-bounded perturbation can be expressed as an \FO transduction, bounded twin-width of $\cf C_0$
would imply bounded twin-width of $\cf C$ by \cite{DBLP:conf/focs/Bonnet0TW20}.
Therefore, if $\cf C$ is of unbounded twin-width, then $\cf C_0$ contains a graph from $\ca R_{\pi}$ 
for every permutation $\pi$ by \cref{thm:intchar}.
To finish, we are going to show that the same holds also for $\cf C$, by constructing a suitable (much larger) permutation $\varrho$
from $\pi$, such that if we take a graph from $\ca R_{\varrho}$, then every $r$-bounded perturbation of it contains a graph from~$\ca R^+_{\pi}$.

So, let $\pi=\pi_1$ be a permutation of $\{1,\ldots,p\}$ for some~$p$ and, same as in the proof of \cref{thm:circlepertu},
let $\pi_2$ be the permutation obtained by concatenation (as an order) of $\pi_1$ followed by the inverse of $\pi_1$ on $\{p+1,\ldots,2p\}$.
We further ``double'' (for technical reasons) $\pi_2$ to form a permutation $\pi_2'$ of the set $T=\{1,\ldots,4p\}$,
such that $\pi_2'$ is determined by the following two linear orders; $1\leq_12\leq_1\ldots\leq_14p$ and
$2\pi_2(1)\!-\!1\leq_22\pi_2(1)\leq_22\pi_2(2)\!-\!1\leq_22\pi_2(2)\leq_2 \ldots\leq_2 2\pi_2(2p)\!-\!1\leq_22\pi_2(2p)$.
Let $U:=T^4$, and let the linear orders $\leq_1$ and $\leq_2$ on $U$ be the standard lexicographic powers of $\leq_1$ and $\leq_2$ on~$T$.
Furthermore, let $Z=U^s$ where $s=2^r$, and $\leq_1$ and $\leq_2$ on $Z$ again be the standard lexicographic powers.
Finally, let $\leq_1$ and $\leq_2$ on $Z$ determine the permutation $\varrho$ of~$Z$.

By the previous, there exists a graph $H^+\in \cf C_0\cap \ca R_{\varrho}$, where~$V(H^+)=Z\cup Z_1\cup Z_2$ where $Z_1\cup Z_2$ are the mates of~$Z$ in it.
Let $H'\in\cf C$ and $X_1,\ldots,X_r$ be the sets to which an $r$-bounded perturbation of $H^+$ has been applied, such that $H'$ is the resulting perturbed graph.
We apply \cref{lem:lexico1} to $Y:=U$ (and to $Z$ and the sets $X_1,\ldots,X_r$); this gives us a homogeneous set $Z_0\subseteq Z$ with some prefix $(a_1,\ldots,a_{\ell-1})$.
Since the linear orders $\leq_1$ and $\leq_2$ on $Z$ have been obtained as lexicographic powers of $\leq_1$ and $\leq_2$ on $Y=U$,
their restriction to $Z_0$ is isomorphic to $\leq_1$ and $\leq_2$ on $U$.
Abusing the notation for simplicity, we hence set $Z_0=U$, and then denote by $U_1$, $U_2$ the corresponding mates of vertices of $U$ within~$H^+$.

Let $H_1=H^+[U\cup U_1\cup U_2]$ be the induced subgraph of $H^+$ on~$U=Z_0$, and likewise $H_1'=H'[U\cup U_1\cup U_2]$.
By homogeneity of $U$ (from \cref{lem:lexico1}), we have that each of the complemented sets $X_i$ either contains $U$ or is disjoint from it.
Consequently, every vertex of $U_1$ either has the same neighborhood to $U$ in $H_1$ as in $H_1'$, or the exactly complementary neighborhood.
Let $X_1'\subseteq U$ be the set of those vertices whose mate in $U_1$ has the same neighborhood to $U$ in $H_1$ as in $H_1'$,
and define analogously $X_2'\subseteq U$ with respect to the mates in~$U_2$.
We apply \cref{lem:lexico1} again, now to $Y:=T$, the set $U$ in place of~$Z$, and to the sets $X_1',X_2'$ ($s=2^2=4$).
From that we get a homogeneous set $U_0\subseteq U$ (with a prefix which is formally appended to the prefix of $Z_0$ above),
and we denote by $T_1\subseteq U_1$ and $T_2\subseteq U_2$ the sets of mates of the vertices of~$U_0$.
Abusing again the notation for simplicity, we set $U_0=T$.

By homogeneity of $T=U_0$, the edge set between $T$ and $T_1$ is either the same in $H_1$ as in $H_1'$, or the exact complement of it.
The same holds for $T$ and~$T_2$.
Now, the $H_1'$-neighborhoods in $T_1$ and $T_2$ define on $T$ linear preorders $\leq_1'$ and $\leq_2'$, respectively (consecutive pairs of $T$ may actually have equal neighborhoods in $T_1$ or $T_2$).
However, with appropriate choice of a representative of every pair of consecutive vertices of $T$ (cf.~the above definition of $\pi_2'$ on $T$), 
giving $T_0\subseteq T$, and similar choice from $T_1$ and $T_2$, we get an induced subgraph $H_1''\subseteq H_1'$, such that
the restricted orders $\leq_1''$ and $\leq_2''$ on $T_0$ (defined analogously by $H_1''$-neighborhoods) give a permutation on $T_0$ isomorphic
to the above permutation $\pi_2$ on $\{1,\ldots,2p\}$ or to its inversion.
In any case, a suborder on $T_0$ defines a permutation isomorphic to $\pi_1$ we started with, and so $H_1'$ contains as an induced subgraph
a member of $\ca R^+_{\pi_1}$.

Since $\pi=\pi_1$ has been chosen arbitrarily, and $\cf C\ni H'$ is hereditary, we conclude that $\cf C$ contains a graph from $\ca R^+_{\pi}$ for every~$\pi$.
\end{proof}

\ifx\proof\inlineproof\else
\begin{proof}[Proof sketch]
If $\cf C$ contains some graph from $\ca R^+_{\pi}$ for every permutation $\pi$, then, by \cref{lem:transduceperm},
the transduced class $\tau(\cf C)$ contains all permutations and hence is of unbounded twin-width, exactly as in the proof of \cref{thm:intchar}.

Conversely, let $\cf C_0$ be a hereditary class of interval graphs, such that every graph in $\cf C$ is 
an $r$-bounded perturbation of a graph in~$\cf C_0$, and that some $r$-bounded perturbation of every graph in $\cf C_0$ falls into~$\cf C$.
Again, since an $r$-bounded perturbation can be expressed as an \FO transduction, bounded twin-width of $\cf C_0$
would imply bounded twin-width of $\cf C$ by \cite{DBLP:conf/focs/Bonnet0TW20}.
Therefore, if $\cf C$ is of unbounded twin-width, then $\cf C_0$ contains a graph from $\ca R_{\pi}$ 
for every permutation $\pi$ by \cref{thm:intchar}.
To finish, we are going to show that the same holds also for $\cf C$, by constructing a suitable (much larger) permutation $\varrho$
from $\pi$, such that if we take a graph from $\ca R_{\varrho}$, then every $r$-bounded perturbation of it contains a graph from~$\ca R^+_{\pi}$.
Details of this construction are left for the appendix.
\end{proof}
\fi

\section{Conclusions}\label{sec:conclu}
%%%%%%%%%%%%%%%%%%%%%%%%%%%%%%%%%%%%%%%%%%%%%%%%%%%%%%%%%%%%%%%%%%%%%%%

We have got precise characterizations of bounded twin-width by explicit obstructions (as induced subgraphs) in the classes
of interval and circle graphs, and in the classes obtained by bounded perturbations from these graphs.
In the case of interval and circle graphs alone, our obstructions also explicitly characterize fixed-parameter tractability
of \FO model checking, under usual complexity assumptions.
While it is relatively easy to relate the bounded twin-width property between a graph class and the class of its bounded perturbations,
since bounded perturbations can be expressed by \FO transductions in both directions, the fact that bounded perturbations do not
change the class of the explicit obstructions is remarkable.
It is perhaps worth to investigate to which level this finding can be generalized.

In order to extend our characterizations of tractability of \FO model checking to bounded perturbations of the classes, we would need the following:
Having a class $\cf C$ of graphs of bounded twin-width with efficiently computable contraction sequences of bounded width,
can we input a graph $G$ which is a bounded perturbation of a graph $G_1\in\cf C$ (without knowing~$G_1$) 
and efficiently compute a contraction sequence of $G$ of bounded width?
We propose this as the main open question of the paper.

\medskip
In another direction, we propose to investigate possible extensions of the results of \cref{sec:interval} to the classes of $k$-thin graphs
\cite{DBLP:journals/dam/BonomoE19} and possibly to $k$-mixed-thin graphs \cite{DBLP:journals/corr/abs-2202-12536}.
While we do not state the formal definition here, we remark that $1$-thin graphs are exactly the interval graphs, and the mentioned classes
naturally generalize interval graphs (and, in particular, proper $k$-mixed-thin graphs are of bounded twin-width~\cite{DBLP:journals/corr/abs-2202-12536}).
An extension of our techniques to these classes seems possible, but not at all trivial.

\bibliography{bibliography,interpretations,k-mixed-thin}

\begin{thebibliography}{10}

\bibitem{DBLP:journals/corr/abs-2202-12536}
Jakub Balab{\'{a}}n, Petr Hlin\v{e}n{\'{y}}, and Jan Jedelsk{\'{y}}.
\newblock Twin-width and transductions of proper $k$-mixed-thin graphs.
\newblock {\em CoRR}, abs/2202.12536, 2022.
\newblock Accepted to WG 2022.

\bibitem{DBLP:journals/corr/abs-2112-08953}
Pierre Berg{\'{e}}, {\'{E}}douard Bonnet, and Hugues D{\'{e}}pr{\'{e}}s.
\newblock Deciding twin-width at most 4 is {NP}-complete.
\newblock {\em CoRR}, abs/2112.08953, 2021.

\bibitem{DBLP:journals/corr/abs-2204-00722}
{\'{E}}douard Bonnet, Dibyayan Chakraborty, Eun~Jung Kim, Noleen K{\"{o}}hler,
  Raul Lopes, and St{\'{e}}phan Thomass{\'{e}}.
\newblock Twin-width {VIII:} delineation and win-wins.
\newblock {\em CoRR}, abs/2204.00722, 2022.

\bibitem{DBLP:conf/soda/BonnetGKTW21}
{\'{E}}douard Bonnet, Colin Geniet, Eun~Jung Kim, St{\'{e}}phan Thomass{\'{e}},
  and R{\'{e}}mi Watrigant.
\newblock Twin-width {II:} small classes.
\newblock In {\em {SODA}}, pages 1977--1996. {SIAM}, 2021.

\bibitem{DBLP:conf/icalp/BonnetG0TW21}
{\'{E}}douard Bonnet, Colin Geniet, Eun~Jung Kim, St{\'{e}}phan Thomass{\'{e}},
  and R{\'{e}}mi Watrigant.
\newblock Twin-width {III:} max independent set, min dominating set, and
  coloring.
\newblock In {\em {ICALP}}, volume 198 of {\em LIPIcs}, pages 35:1--35:20.
  Schloss Dagstuhl - Leibniz-Zentrum f{\"{u}}r Informatik, 2021.

\bibitem{DBLP:journals/corr/abs-2102-03117}
{\'{E}}douard Bonnet, Ugo Giocanti, Patrice~Ossona de~Mendez, and St{\'{e}}phan
  Thomass{\'{e}}.
\newblock Twin-width {IV:} low complexity matrices.
\newblock {\em CoRR}, abs/2102.03117, 2021.

\bibitem{DBLP:conf/soda/BonnetKRT22}
{\'{E}}douard Bonnet, Eun~Jung Kim, Amadeus Reinald, and St{\'{e}}phan
  Thomass{\'{e}}.
\newblock Twin-width {VI:} the lens of contraction sequences.
\newblock In {\em {SODA}}, pages 1036--1056. {SIAM}, 2022.

\bibitem{DBLP:conf/focs/Bonnet0TW20}
{\'{E}}douard Bonnet, Eun~Jung Kim, St{\'{e}}phan Thomass{\'{e}}, and
  R{\'{e}}mi Watrigant.
\newblock Twin-width {I:} tractable {FO} model checking.
\newblock In {\em {FOCS}}, pages 601--612. {IEEE}, 2020.

\bibitem{DBLP:journals/corr/abs-2102-06880}
{\'{E}}douard Bonnet, Jaroslav Ne\v{s}et\v{r}il, Patrice~Ossona de~Mendez,
  Sebastian Siebertz, and St{\'{e}}phan Thomass{\'{e}}.
\newblock Twin-width and permutations.
\newblock {\em CoRR}, abs/2102.06880, 2021.

\bibitem{DBLP:journals/dam/BonomoE19}
Flavia Bonomo and Diego de~Estrada.
\newblock On the thinness and proper thinness of a graph.
\newblock {\em Discret. Appl. Math.}, 261:78--92, 2019.

\bibitem{recogIntervalLinear}
Kellogg~S. Booth and George~S. Lueker.
\newblock Testing for the consecutive ones property, interval graphs, and graph
  planarity using {PQ}-tree algorithms.
\newblock {\em J. Comput. Syst. Sci.}, 13(3):335--379, 1976.

\bibitem{DBLP:journals/combinatorica/Bouchet87}
Andr{\'{e}} Bouchet.
\newblock Reducing prime graphs and recognizing circle graphs.
\newblock {\em Comb.}, 7(3):243--254, 1987.

\bibitem{DBLP:conf/dmtcs/DowneyFT96}
R.~G. Downey, M.~R. Fellows, and U.~Taylor.
\newblock The parameterized complexity of relational database queries and an
  improved characterization of {W[1]}.
\newblock In {\em First Conference of the Centre for Discrete Mathematics and
  Theoretical Computer Science, {DMTCS} 1996}, pages 194--213. Springer-Verlag,
  Singapore, 1996.

\bibitem{DBLP:conf/focs/GajarskyHLOORS15}
Jakub Gajarsk{\'{y}}, Petr Hlin\v{e}n{\'{y}}, Daniel Lokshtanov, Jan
  Obdr\v{z}{\'{a}}lek, Sebastian Ordyniak, M.~S. Ramanujan, and Saket Saurabh.
\newblock {FO} model checking on posets of bounded width.
\newblock In {\em {FOCS}}, pages 963--974. {IEEE} Computer Society, 2015.

\bibitem{DBLP:journals/corr/GajarskyHLOORS15}
Jakub Gajarsk{\'{y}}, Petr Hlin\v{e}n{\'{y}}, Daniel Lokshtanov, Jan
  Obdr\v{z}{\'{a}}lek, Sebastian Ordyniak, M.~S. Ramanujan, and Saket Saurabh.
\newblock {FO} model checking on posets of bounded width.
\newblock {\em CoRR}, abs/1504.04115, 2015.

\bibitem{ghkost13}
R.~Ganian, P.~Hlin\v{e}n{\'{y}}, D.~Kr{\'{a}l'}, J.~Obdr\v{z}{\'{a}}lek,
  J.~Schwartz, and J.~Teska.
\newblock {FO} model checking of interval graphs.
\newblock In {\em {ICALP} 2013, Part {II}}, volume 7966 of {\em LNCS}, pages
  250--262. Springer, 2013.

\bibitem{gks14}
M.~Grohe, S.~Kreutzer, and S.~Siebertz.
\newblock Deciding first-order properties of nowhere dense graphs.
\newblock {\em J. {ACM}}, 64(3):17:1--17:32, 2017.

\bibitem{DBLP:journals/comgeo/HlinenyPR19}
Petr Hlin\v{e}n{\'{y}}, Filip Pokr{\'{y}}vka, and Bodhayan Roy.
\newblock {FO} model checking on geometric graphs.
\newblock {\em Comput. Geom.}, 78:1--19, 2019.

\bibitem{McCartyRose2021}
{Rose McCarty}.
\newblock {\em Local Structure for Vertex-Minors}.
\newblock PhD thesis, University of Waterloo, 2021.
\newblock URL: \url{http://hdl.handle.net/10012/17633}.

\bibitem{seese96}
D.~Seese.
\newblock Linear time computable problems and first-order descriptions.
\newblock {\em Math. Structures Comput. Sci.}, 6(6):505--526, 1996.

\bibitem{DBLP:journals/jal/Spinrad94}
Jeremy~P. Spinrad.
\newblock Recognition of circle graphs.
\newblock {\em J. Algorithms}, 16(2):264--282, 1994.

\end{thebibliography}

\end{document}